\documentclass[conference]{IEEEtran}
\usepackage[noadjust]{cite}
\usepackage{graphicx,color}
\usepackage[dvipsnames]{xcolor}
\usepackage{amsmath,amsbsy,amsfonts,amssymb,amsthm}
\usepackage{mathrsfs,bm}
\usepackage{mathtools}
\usepackage{flushend}
\usepackage{lipsum}

\mathtoolsset{showonlyrefs}
\ifCLASSOPTIONcompsoc
\usepackage[caption=false,font=normalsize,labelfont=sf,textfont=sf]{subfig}
\else
\usepackage[caption=false,font=footnotesize]{subfig}
\fi

\IEEEoverridecommandlockouts % To allow \thanks{} command to work

\usepackage{etoolbox}
\usepackage{tikz}
\newrobustcmd*{\squareA}[1]{\tikz{\filldraw[draw=#1,fill=#1] (0,-0)
rectangle (0.1cm,0.14cm);}}
\newrobustcmd*{\mycircle}[1]{\tikz{\filldraw[draw=#1,fill=#1] (0,-0.3) circle [radius=0.08cm];}}

\newcommand{\figr}{Fig.~}
\newcommand{\secr}{Sec.~}

\usepackage{acro}
\DeclareAcronym{AWGN}{short = AWGN ,long = additive white gaussian noise}
\DeclareAcronym{AoI}{short = AoI ,long = age of information}
\DeclareAcronym{AoII}{short = AoII ,long = age of incorrect information}

\DeclareAcronym{CDF}{short = CDF ,long = cumulative distribution function}
\DeclareAcronym{CRA}{short = CRA ,long = contention resolution ALOHA}
\DeclareAcronym{CRDSA}{short = CRDSA ,long = contention resolution diversity slotted ALOHA}
\DeclareAcronym{CSA}{short = CSA ,long = coded slotted ALOHA}
\DeclareAcronym{C-RAN}{short = C-RAN ,long = cloud radio access network}
\DeclareAcronym{DAMA}{short = DAMA ,long = demand assigned multiple access}
\DeclareAcronym{DSA}{short = DSA ,long = diversity slotted ALOHA}
\DeclareAcronym{eMBB}{short = eMBB ,long = enhanced mobile broadband}
\DeclareAcronym{FEC}{short = FEC ,long = forward error correction}
\DeclareAcronym{GEO}{short = GEO ,long = geostationary orbit}
\DeclareAcronym{GF}{short = GF ,long = generating function}
\DeclareAcronym{IC}{short = IC ,long = interference cancellation}
\DeclareAcronym{IoT}{short = IoT ,long = Internet of things}
\DeclareAcronym{IRSA}{short = IRSA ,long = irregular repetition slotted ALOHA}
\DeclareAcronym{LEO}{short = LEO ,long = low Earth orbit}
\DeclareAcronym{MAC}{short = MAC ,long = medium access}
\DeclareAcronym{mMTC}{short = mMTC ,long = massive machine-type communications}
\DeclareAcronym{MC}{short = MC ,long = Markov chain}
\DeclareAcronym{PDF}{short = PDF ,long = probability density function}
\DeclareAcronym{PER}{short = PER ,long = packet error rate}
\DeclareAcronym{PLR}{short = PLR ,long = packet loss rate}
\DeclareAcronym{PMF}{short = PMF ,long = probability mass function}
\DeclareAcronym{RA}{short = RA ,long = random access}
\DeclareAcronym{rv}{short = r.v. ,long = random variable}
\DeclareAcronym{SA}{short = SA , long = slotted ALOHA}
\DeclareAcronym{SIC}{short = SIC ,long = successive interference cancellation}
\DeclareAcronym{SNR}{short = SNR ,long = signal-to-noise ratio}
\DeclareAcronym{SFG}{short = SFG ,long = signal flow graph}
\DeclareAcronym{TDM}{short = TDM ,long = time division multiplexing}
\begin{document}

\title{\huge On the Uncertainty of a Simple Estimator for Remote Source Monitoring over ALOHA Channels}
\author{
\IEEEauthorblockN{Andrea Munari\\
\IEEEauthorblockA{Institute of Communications and Navigation, German Aerospace Center (DLR), Wessling, Germany
}}
}
\date{}
\maketitle
\thispagestyle{empty}
\pagestyle{empty}

\begin{abstract}
    Efficient remote monitoring of distributed sources is essential for many Internet of Things (IoT) applications. This work studies the uncertainty at the receiver when tracking two-state Markov sources over a slotted random access channel without feedback, using the conditional entropy as a performance indicator, and considering the last received value as current state estimate. We provide an analytical characterization of the metric, and evaluate three access strategies: (i) maximizing throughput, (ii) transmitting only on state changes, and (iii) minimizing uncertainty through optimized access probabilities. Our results reveal that throughput optimization does not always reduce uncertainty. Moreover, while reactive policies are optimal for symmetric sources, asymmetric processes benefit from mixed strategies allowing transmissions during state persistence.
\end{abstract}

\newtheorem{prop}{Proposition}
\newtheorem{lemma}{Lemma}
\newtheorem{remark}{Remark}

% for Stirling number
\DeclareRobustCommand{\stirling}{\genfrac\{\}{0pt}{}}

% maths and probability
\newcommand{\pr}{\ensuremath{\mathsf P}}
\newcommand{\expOp}{\ensuremath{\mathbb E}}
\newcommand{\de}{\mathrm{d}}
\newcommand{\given}{\, | \,}

% Markov chains
\newcommand{\pmc}{\ensuremath{q}}

\newcommand{\Mc}{\ensuremath{X}}
\newcommand{\Mcn}{\ensuremath{X_n}}
\newcommand{\Mcni}{\ensuremath{X_n^{(i)}}}
\newcommand{\Est}{\ensuremath{\hat{X}}}
\newcommand{\Estn}{\ensuremath{\hat{X}_n}}
\newcommand{\mc}{\ensuremath{x}}
\newcommand{\mcn}{\ensuremath{x_n}}
\newcommand{\mcni}{\ensuremath{x_n^{(i)}}}
\newcommand{\est}{\ensuremath{\hat{x}}}
\newcommand{\estn}{\ensuremath{\hat{x}_n}}
\newcommand{\Irt}{\ensuremath{W}}
\newcommand{\irt}{\ensuremath{w}}

% channel access
\newcommand{\pTx}{\ensuremath{\lambda}}
\newcommand{\pTxVec}{\ensuremath{\bm\pTx}}
\newcommand{\avgPTx}{\ensuremath{\overbar\lambda}}
\newcommand{\avgAct}{\ensuremath{\rho}}
\newcommand{\ps}{\ensuremath{\mathsf{p_s}}}

% transition and stationary probabilities
\newcommand{\qZO}{\ensuremath{\alpha}}
\newcommand{\qOZ}{\ensuremath{\beta}}
\newcommand{\asymm}{\ensuremath{\eta}}
\newcommand{\statZ}{\ensuremath{\pi_0}}
\newcommand{\statO}{\ensuremath{\pi_1}}
\newcommand{\statZZ}{\ensuremath{\pi_{0,0}}}
\newcommand{\statOO}{\ensuremath{\pi_{1,1}}}
\newcommand{\statZO}{\ensuremath{\pi_{0,1}}}
\newcommand{\statOZ}{\ensuremath{\pi_{1,0}}}

\newcommand{\nodes}{\ensuremath{\mathsf m}}

% basic metrics
\newcommand{\tru}{\ensuremath{\mathsf S}}
\newcommand{\load}{\ensuremath{\mathsf G}}
\newcommand{\Agen}{\ensuremath{\Delta_n}}
\newcommand{\agen}{\ensuremath{\delta_n}}
\newcommand{\ent}{\ensuremath{H}}
\newcommand{\condent}{\ensuremath{\mathsf h}}

% AoI related metrics
\newcommand{\overbar}[1]{\mkern 1.5mu\overline{\mkern-3mu#1\mkern0.5mu}\mkern 1.5mu}
\newcommand{\overbara}[1]{\mkern 1.5mu\overline{\mkern0.1mu#1\mkern-0.1mu}\mkern 1.5mu}

\section{Introduction} \label{sec:intro}

Remote observation of distributed sources is key to many Internet of things (IoT) applications, ranging from agricultural and environmental monitoring, to industrial control and asset tracking \cite{Bedewy21_TIT,Yates20_Survey,Uysal21_Semantic}. In these contexts, a possibly large number of devices sample physical processes of interest, and deliver updates to a common receiver for further processing or decision making \cite{Pappas24_TCOM,Soleymani20_valueInfo}. To accommodate the sporadic and potentially unpredictable traffic generated by low-complexity terminals over a shared wireless channel, random access schemes based on variations of the ALOHA policy \cite{Abramson77:PacketBroadcasting} are typically employed, as epitomized by several commercial solutions \cite{LoRa,SigFox}.

The relevance of these applications has steered recent research towards understanding how the channel access policies shall be tuned, aiming to provide the monitor with an accurate estimate of the state of tracked processes. In this perspective, first steps were taken tackling freshness, captured by means of age of information (AoI) \cite{Yates17:AoI_SA,Modiano18_AoI,Munari21_TCOM_AoI,Uysal21_AlohaThresh,Bidokhti22_TIT,Badia22_NetwLetters,Badia24_TMC,Liew20_INFOCOM,Munari23_TCOM}. These studies were later refined to account for the actual status of the sources and the estimate available at the receiver, e.g., via age of incorrect information \cite{Ephremides19_AoII,Chiariotti23:ICC,Munari24_ICC,Chiariotti25_INFOCOM}, false alarm and missed detection probabilities \cite{Pappas24_MOBIHOC,Munari23_Asilomar}, as well as information-theoretic inspired \cite{Cocco23_JSAIT,Liew22_TIT,Pappas23_WiOpt} and other metrics \cite{Soleymani20_valueInfo,Pappas24_WiOpt,Yates20_Survey}.

In this work, we take the lead from two practical observations. Fist, many IoT networks operate without feedback, so that nodes are unaware of the outcome of their transmissions, e.g., \cite{LoRa}. Second, we note that control and monitoring tasks are typically performed at the application layer, which, in turn, may receive from lower layers a limited amount of information. In most systems, indeed, only messages generated by the tracked source would be forwarded to the application, possibly together with general information on the network configuration, e.g., number of registered nodes. This configuration hinders the realization of advanced estimators, e.g., \cite{Cocco23_JSAIT}, which exploit cross-layer information such as level of interference, lack of transmissions or reception of messages from other sources, to refine knowledge on the process. 

Based on this, we consider a basic solution, which simply retains the last received status update as current estimate \-- sometimes also referred to as martingale estimator \cite{Ulukus24_Martingale,Munari23_Asilomar}. Leaning on this, we study the problem of monitoring two-state Markov sources over a slotted random access channel without feedback, and characterize performance in terms of the uncertainty at the receiver given the available information, i.e., the last obtained message on the state of a process of interest and when such value was measured, corresponding to the AoI. Specifically, we focus on variations of ALOHA in which each terminal probabilistically transmits an update based on the current and previous value of the tracked process. We provide an analytical derivation of the conditional entropy at the receiver, and compare the behavior of three main strategies: a benchmark scheme where access probabilities are set to maximize throughput, regardless of the source evolution; a reactive approach in which nodes access the medium only when a change in source state takes place; and a balanced strategy, whose access probabilities are obtained to minimize the uncertainty metric. Our study reveals non-trivial insights, showing how optimizing channel throughput does not necessarily lead to a lower uncertainty at the receiver. Moreover, we show that, while for symmetric sources (i.e., with balanced transitions between the two states), a reactive policy is optimal, more complex solutions that also foresee transmission in the absence of source transitions are to be preferred when asymmetric processes are monitored.

\emph{Notation}: Matrices are indicated in capital boldface, e.g. $\mathbf A$, whereas column (row) vectors in lowercase boldface, e.g. $\mathbf v$ ($\mathbf v^{\mathsf T}$). We denote a discrete random variable (r.v.) and its realization by lower- and uppercase letters, e.g., $X$ and $x$. Its probability mass function is indicated as \mbox{$\mathsf P(X=x) = p(x)$}, and the conditional distribution of $X$ given $Y$ is expressed as $p(x\given y)$. For a homogeneous discrete time Markov chain $X_n$, $n\in\mathbb N$, with finite space state $\mathcal X$, $q_{ij}$ is the one-step transition probability from state $i$ to $j$, whereas $q_{ij}(\ell)$ is the $\ell$-step transition probability. Finally, $\mathbf 1_\ell$ is the $\ell$ dimensional column vector with all ones, and $\mathbf I_\ell$ the $\ell\times\ell$ identity matrix.%——, i.e., the probability of being in $j$ at step $\ell$ starting in $i$ at time $0$.%Let $\mathbf P$ be the $k\times k$ $\ell$-step transition matrix of the Markov chain $X_n$, and let $x$ be the $i$-th element of the ordered state space $\mathcal X$. Then, we denote as $\mathbf e_i$ the $k\times 1$ column vector composed of all zeros, except for a $1$ the in the $i$-th position.  

\section{System Model}
\label{sec:sysModel}

We study a wireless network composed of \nodes\ terminals (nodes). Time is divided in slots, and each node monitors an independent discrete-time, two-state, Markov chain taking values in $\{0,1\}$. We denote by $\Mcni$, $n\in\mathbb N$, the chain observed by terminal $i$. 
%\begin{figure}
%    \centering
%    \includegraphics[width=.8\columnwidth]{./Figures/source_MC.pdf}
%    \caption{Discrete time Markov chain \mcn\ describing the evolution of one of the monitored sources.}
%    \label{fig:mc_source}
%\end{figure}
At the start of a slot, each process transitions between its states following the one-step probabilities reported in \figr\ref{fig:markovChains}a. For convenience, we introduce the \emph{asymmetry factor} $\asymm := \qZO/\qOZ$, capturing the ratio of the average time spent in state $1$ (i.e., $1/\qOZ$) with respect to state $0$ (i.e., $1/\qZO$). When $\asymm=1$, we speak of symmetric sources. The stationary distribution of the chain follows as $\statZ = \qOZ/(\qZO+\qOZ)$ and $\statO = 1-\statZ$.

Nodes share a wireless channel, and aim at reporting the state of the monitored sources to a common receiver (sink). Specifically, at the start of a slot, each terminal independently decides whether to transmit a packet, containing the current state of the Markov chain it observes. Accordingly, a slotted ALOHA protocol is implemented. No feedback is provided by the sink, and no retransmissions are performed by nodes. Following the well-established collision channel model \cite{Abramson77:PacketBroadcasting}, we assume that a slot over which two or more packets are sent (collision) does not allow retrieval of information at the sink, whereas successful decoding takes place whenever a single terminal transmits during a slot. 

In the remainder, we will focus on access schemes in which the transmission probability of a terminal is dictated by the previous and present state of the monitored process. We denote this quantity as $\pTx_{\mc_{n-1} \mcn}$, so that the access policy is fully specified by the vector $\pTxVec = [\pTx_{00},\pTx_{01},\pTx_{10},\pTx_{11}]$. We remark that, while $\pTxVec$ is the same for all nodes, the actual contention probability of each of them depends on the current evolution of its source. In view of this, the characterization of the number of terminals accessing the channel over a slot would require to jointly track all the Markov processes. Albeit conceptually viable, this approach soon becomes cumbersome as \nodes\ grows. We thus resort to an approximation, whose tightness is discussed in Sec.~\ref*{sec:results}, and model the success probability for a transmitted packet as
\begin{align}
    \ps := (1-\avgPTx)^{\nodes-1}
    \label{eq:ps}
\end{align}
where we have introduced the ancillary quantity \mbox{$\avgPTx := \statZ [ (1{-}\qZO) \pTx_{00} {+} \qZO \pTx_{01} ] + \statO [\qOZ \pTx_{10} {+} (1{-}\qOZ) \pTx_{11}]$}.
%\begin{align}
%    \avgPTx := \statZ [\, (1-\qZO) \pTx_{00} + \qZO \pTx_{01} \,] + \statO [\,\qOZ \pTx_{10} + (1-\qOZ) \pTx_{11}\,].
%\end{align}
In other words, we consider an i.i.d. behavior for all the $\nodes-1$ terminals other than the sender of interest, assuming that each of them transmits with probability \avgPTx. This, in turn, captures the average access probability for a node in stationary conditions.

\begin{figure}    
    \subfloat[]{
        \includegraphics[width=.25\columnwidth]{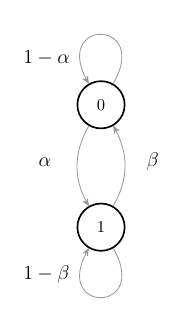}}
    \hspace*{.1em}
    \subfloat[]{
        \includegraphics[width=.7\columnwidth]{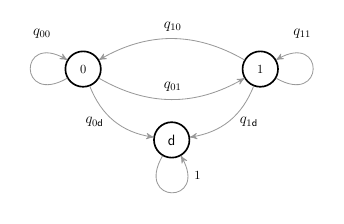}}
    \caption{(a) Markov chain \Mcn\ describing a monitored source; (b) terminating Markov chain $Y_n$ used to characterize $\ent(\Mcn\given\Agen,\Estn)$. The process enters the absorbing state $\mathsf d$ when an update from the reference node is received. Conversely, it moves between $0$ and $1$, describing the corresponding current source value, so long as no update message from the reference terminal arrives.}
    \label{fig:markovChains}
    \vspace{-1em}
\end{figure}

Within our system, the sink maintains an estimate $\Estn^{(i)}$ of the state of each monitored process. To this aim, we consider a simple solution, updating the estimate every time a message containing the current state of the source is received, and retaining the previous knowledge otherwise. Formally:
\begin{align}
    \!\!\!\!\Estn^{(i)} = 
    \begin{split}
    \begin{cases}
        \hspace*{.3em} \Mcni            \!\!\!& \text{ if node $i$ delivers message at slot $n$}\\[.3em]
        \hspace*{.3em} \Est_{n-1}^{(i)} \!\!\!& \text{ otherwise}
    \end{cases}
    \end{split}
    \label{eq:dh}
\end{align} 
with $\Estn^{(i)}\in\{0,1\}$. Without loss of generality, we will focus on the behavior of a reference source, and drop superscript $i$.
%\begin{remark}
%    \emph{We note that the presented approach requires no knowledge of network cardinality, employed access policy and source statistics. The estimator is only triggered upon successful decoding of a message, and can simply be run at the application layer, without the need for any other cross-layer information exchange (e.g., presence of a collided or idle slot). As such it may be of particular relevance for a large number of practical/already deployed IoT use cases.}
%\end{remark}
 
For the setting under study, we want to characterize the uncertainty of the sink on the state of the reference process. To this aim, we note that, at time $n$, a receiver implementing the estimator in \eqref{eq:dh} only has knowledge about (i) the last received update from the node of interest, i.e., \Estn, and (ii) the time elapsed since such message was retrieved. We denote the random process describing the latter as $\Agen \in \mathbb N_0$, and observe that \Agen\ corresponds the current AoI \cite{Yates20_Survey} at the sink. Indeed, each sent packet contains up to date information on the monitored source, %implementing a generate-at-will model, 
and $\Agen$ is exactly the difference between the current time and the time stamp of the last received message. We assume that, upon reception, \Agen\ is reset to $0$.

A natural measure of the sink uncertainty at time $n$, for a given AoI-estimate pair $(\agen,\estn)$, is thus given by the entropy
\begin{align}
    \mathsf h(\agen,\estn) := \ent( \Mcn \,|\, \Agen = \agen, \Estn = \estn).
    \label{eq:cond_ent}
\end{align}
An example of the time evolution of $\condent(\agen,\estn)$ is reported in \figr\ref{fig:timeline}. %In this case, $\nodes=50$ nodes were considered, implementing a transmission policy $\pTx_{\mc_{n-1} \mcn} = 1/ \nodes$, for any $(\mc_{n-1},\mcn)$ pair, and the source parameters were set as $\qZO=0.1$, $\qOZ=0.01$. 
The metric is reset to zero each time a message from the reference node is decoded. Instead, in the absence of updates, $\condent(\agen,\estn)$ tends to the stationary entropy of the source, $\mathsf H(X) = -\statZ \log_2 \statZ -\statO \log_2 \statO$. Finally, the peak shown in the plot denotes the higher uncertainty at the sink  experienced following reception of a message notifying of a source transition to the less likely state $0$. 

Leaning on this notation, we evaluate the average performance of the system in terms of the conditional entropy
\begin{align}
    \ent(\Mcn\,|\, \Agen,\Estn) = \sum_{\substack{\agen \in \mathbb N_0 \\ \estn \in \{0,1\} } } p(\agen,\estn)\, \condent(\agen,\estn).
    \label{eq:ent_def}
\end{align}
%with $\agen \in \mathbb N_0$, and $\estn \in \{0,1\}$.

\vspace{.5em}
\textbf{Remark.} \emph{The considered estimator is only updated upon successful decoding of a message from the source of interest, and can be run at the application layer, without the need for any other cross-layer information exchange (e.g., presence of a collided or idle slot). %As such it may be of particular relevance for a large number
    %of practical/already deployed IoT use cases.  
    In the remainder, we will provide a framework to understand, as protocol designers, how the medium contention shall be tuned  to minimize the average uncertainty $\ent(\Mcn\given \Agen,\Estn)$. On the other hand, as clarified in \secr\ref{sec:analysis}, knowledge of network cardinality, access parameters, and source statistics, allows an application to base control and decisions on its current uncertainty computed via \eqref{eq:cond_ent}.}
    %On the other hand, as clarified     in Sec. III, knowledge of network cardinality, access parameters, and source statistics suffice to compute (3), allowing an 
    %application with such knowledge to base control and decisions on its current uncertainty

    \begin{figure}
        \centering
        \includegraphics[width=.9\columnwidth]{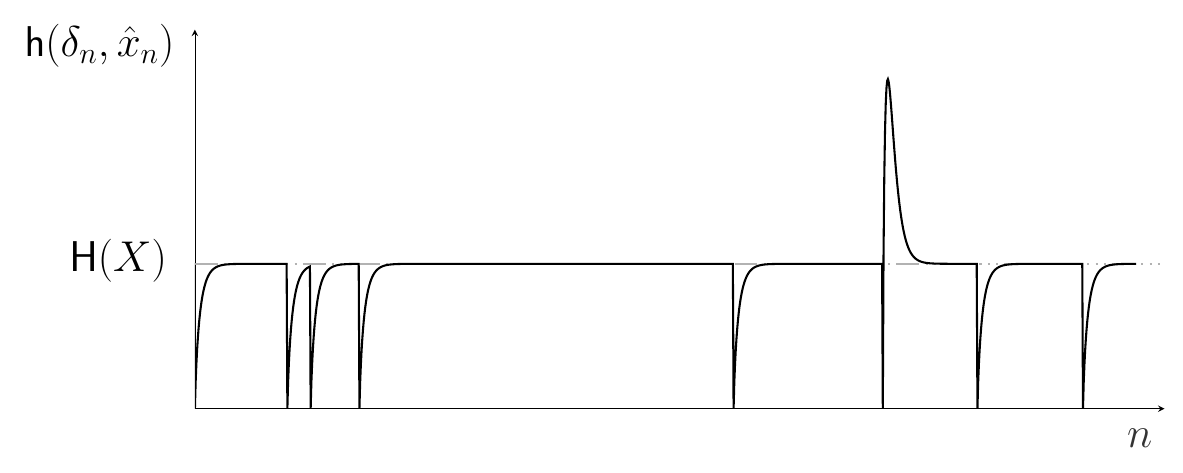}
        \caption{Example of time evolution of the entropy $\condent(\agen,\estn)$. In this case, $\qZO=0.1$, $\qOZ=0.01$, $\nodes=50$, $\pTx_{\mc_{n-1}\mcn} = 1/\nodes$, $\forall \, (\mc_{n-1},\mcn)$.}
        \vspace{-1em}
        \label{fig:timeline}
    \end{figure}
\section{Analysis}
\label{sec:analysis}

%Let us start by considering the entropy
%\begin{align}
%    \condent(\agen,\estn) = -\sum_{\mcn} p(\mcn\given\agen,\estn) \log_2 p(\mcn\given\agen,\estn)
%    \label{eq:cond_entropy_formula}
%\end{align}
%whose calculation requires the conditional distribution of the  source state given the current receiver estimate and the AoI value. 
To characterize the uncertainty at the receiver, we will resort to the auxiliary terminating Markov chain $Y_n$ reported in \figr\ref{fig:markovChains}b, with state-space $\mathcal Y = \{0,1,\mathsf d\}$. The chain transitions between the two upper states so long as the sink receives no message from the reference node, with $0$ and $1$ denoting the actual current source value. In turn, the process enters the absorbing state $\mathsf d$ as soon as an update refreshing the receiver estimate is delivered. 
%\begin{figure}
%    \centering
%    \includegraphics[width=.8\columnwidth]{./Figures/absorbing_MC.pdf}
%    \caption{Terminating Markov chain useful to characterize the distributions needed in the computation of the conditional entropy $\ent(\Mcn\given\Agen,\Estn)$. The chain transitions to the absorbing state $\mathsf d$ whenever an update from the reference node is received, refreshing the estimate. Conversely, the chain moves between states $0$ and $1$, describing the corresponding current reference source value, so long as no update message from the reference terminal arrives.}
%    \label{fig:absMC}
%\end{figure}
The one-step transition matrix for the process can be written as
\begin{align}
    \mathbf P = 
    \left(
        \begin{array}{cc|c}
            q_{00}  & q_{01} & q_{0\mathsf d}\\  
            q_{10}  & q_{11} & q_{1\mathsf d}\\[.3em]
            \hline
            0       & 0      & 1
        \end{array}
    \right)
     =
    \begin{pmatrix}
        \mathbf A & \mathbf a_{\mathsf d}\\
        \mathbf 0 & 1 \\
    \end{pmatrix}
\end{align}
where $\mathbf A$ is the $2\times 2$ matrix that captures transitions between $0$ and 1, and the $2\times 1$ vector $\mathbf a_{\mathsf d}$ contains the probability of being absorbed from each of the two states. In turn, the transition probabilities can be derived for the considered system model. For instance, focusing on state $0$, the chain will move to $1$ with probability $q_{01} = \qZO (1-\pTx_{01}\ps)$. Here, the first factor captures the fact that the source has to move to state $1$, whereas the second accounts for the lack of an update delivery from the terminal over the current slot (which would lead to absorption). Similarly, the chain remains in $0$ with probability \mbox{$q_{00}=(1-\qZO)(1-\pTx_{00}\ps)$}. Finally, if a message is successfully sent, the chain moves to $\mathsf d$, regardless of the state of the source, i.e., with overall probability $q_{0\mathsf d} = (\qZO\pTx_{01}+(1-\qZO)\pTx_{00})\ps$. The transitions from $1$ are immediately derived in the same manner and are not reported for the sake of compactness. 

The chain leads to a first result, captured in the following
\begin{prop} \label{prop1}The conditional probability of the reference source being in state \mcn\ given that its current AoI is \agen\ and the last received updated contained state \estn\ is 
    \begin{align}
        p(\mcn\given\agen,\estn) = \frac{ \mathbf e_{\estn}^{\mathsf T} \mathsf A^{\agen} \,\mathbf e_{\mcn}}{\mathbf e_{\estn}^{\mathsf T} \mathsf A^{\agen} \mathbf 1_2}
        \label{eq:pmfXnGivenDeltaXnHat}
    \end{align}
    where, for any $\estn$ and \mcn\ in  $\mathcal X$, $\mathbf e_{\estn}$ and $\mathbf e_{\mcn}$ are defined as \mbox{$\mathbf e_0 = [1,0]^{\mathsf T}$}; $\mathbf e_1 = [0,1]^{\mathsf T}$.
\end{prop}
%\begin{align}
%    p(\mcn\given\agen,\estn) = \frac{[\mathsf A^{\agen}]_{\estn,\mcn}}{[\mathsf A^{\agen}]_{\estn,0} + [\mathsf A^{\agen}]_{\estn,1}}
%    \label{eq:pmfXnGivenDeltaXnHat}
%\end{align}
\begin{proof}
We observe that, for the Markov chain of \figr\ref{fig:markovChains}b, the $\ell$-step transition probability from $i$ to $j$, with $i,j \in \{0,1\}$, provides the joint distribution of the source being in state $j$ and of not having delivered an update over the last $\ell \geq 1$ slots, given that its state $\ell$ slots ago was $i$. This corresponds to having an estimate $i$, content of the last received message, and a current AoI value $\ell$. By the definition of conditional probability, the sought distribution follows as
\begin{align}
    p(\mcn\given\agen,\estn) = \frac{q_{\estn \mcn}(\agen)}{p(\agen\given\estn)}.
\end{align}
The numerator is directly given by the \estn-row, \mcn-column element of the \agen-step transition matrix $\mathbf P^\agen$ of the chain. By the structure of $\mathbf P$, it is immediate to verify that this corresponds to $\mathbf e_{\estn}^{\mathsf T} \mathsf A^{\agen} \,\mathbf e_{\mcn}$. On the other hand, %the conditional probability of having AoI value $\agen$ given a last update of value $\estn$ 
$p(\agen\given\estn)$ can be derived as the probability of the chain not to transition to state $\mathsf d$ for \agen\ steps having started in \estn. This evaluates to $\sum\nolimits_{y\in\mathcal Y\setminus\{\mathsf d\}} q_{\estn y}(\agen) = q_{\estn 0}(\agen)  + q_{\estn 1}(\agen)$, where both addends are again obtained as elements of matrix $\mathbf A^{\agen}$.
\end{proof}
%where 
%\begin{align}
%    q_{\estn\mcn}(\agen) = \mathbf e_{\estn}^{\mathsf T} \,\mathbf P^{\agen} \,\mathbf e_{\mcn}.
%\end{align}
The result allows then to evaluate the performance at the receiver given the current conditions in terms of AoI and estimate, resorting to the definition in \eqref{eq:cond_ent}.
\begin{lemma}
    The receiver uncertainty on \Mcn, given \agen\ and \estn\ can be computed for any channel access strategy $\bm \lambda$ as
        \begin{align}
        \condent(\agen,\estn) = -\sum\nolimits_{\mcn} p(\mcn\given\agen,\estn) \log_2 p(\mcn\given\agen,\estn)
    \end{align}
    where $p(\mcn\given\agen,\estn)$ is obtained through \eqref{eq:pmfXnGivenDeltaXnHat}.
\end{lemma}

Let us now turn our attention to the derivation of the conditional entropy $\ent(\Mcn\given\Agen,\Estn)$ in \eqref{eq:ent_def}, which further requires the joint distribution of the current AoI and estimate available at the receiver at a general time slot $n$, i.e., $p(\agen,\estn) = p(\agen\given\estn) p(\estn)$. In the following, we streamline the steps for its computation through Prop. \ref{prop:condXn} and \ref{prop:statEst}. In turn, these results require a preliminary characterization of the inter-refresh time. Specifically, let us denote by \Irt\ the stochastic process describing the duration between two successive message receptions from the reference source, i.e., between two estimate updates at the sink. With this definition, we have:
\begin{prop}
Conditioned on the current estimate available at the receiver, the probability distribution of the process $W$ and its expected value are given by
\begin{align}
    p(\irt\given\estn) = \mathbf e_{\estn}^{\mathsf T} \mathbf A^{\irt-1} \, \mathbf a_{\mathsf d}
    \label{eq:condPMfW}
\end{align}
\begin{align}
    \mathbb E[\Irt \given \Estn=\estn] = \mathbf e_{\estn}^{\mathsf T} (\mathbf I_2 - \mathsf A)^{-1} \, \mathbf 1_2
    \label{eq:avgW}
\end{align}
where $\mathbf e_{\estn}$, $\estn\in\mathcal X$, is defined as in Prop.\ref{prop1}.
\end{prop}
\begin{proof}
    The results follows by observing that the distribution of \Irt, conditioned on the period being characterized by an estimate value $\hat{x}$ at the receiver, corresponds to the absorption time for the auxiliary chain in \figr\ref{fig:markovChains}b when starting from $\hat{x}$. Note indeed that, counting the steps to absorption starting from state $i\in\{0,1\}$ is equivalent to assuming reception of a message at time $0$ \--- stating that the reference source is in state $i$ \---, and thus starting a period over which the sink will keep $i$ as estimate. In turn, the distribution of the absorption time can be obtained using standard methods for Markov chains \cite{Kemeny76}, leading to the discrete phase-type distribution reported in \eqref{eq:condPMfW}, and the corresponding average absorption time in \eqref{eq:avgW}.        
\end{proof}

The proposition allows to derive the statistics of the current AoI value and the stationary distribution of the estimate.
\begin{prop}\label{prop:condXn}
    Conditioned on the estimate available at the receiver, the current AoI follows the probability distribution
    \begin{align}
        p(\agen\given\estn) = \frac{1}{\mathbf e_{\estn}^{\mathsf T} (\mathbf I_2 - \mathsf A)^{-1} \, \mathbf 1_2} \cdot \sum_{w>\agen} p(w\given \estn).
        \label{eq:condPMFAge_complete}
    \end{align}
    \begin{proof}
        For a generic time instant $n$, let us indicate as $\Irt(n)$ the duration of the estimate inter-refresh period $n$ falls into. The probability that $\Irt(n)$ lasts for $w$ slots follows as
        \begin{align}
        \mathsf P \!\left( \Irt(n)=\irt \given \Estn=\estn \right) = \frac{\irt \, p(\irt\given\estn)}{\sum\nolimits_\irt \irt \, p(\irt\given\estn)}
        \label{eq:condProbW}
        \end{align}
        capturing the fraction of time spent by the system in estimate inter-refresh periods of duration $\irt$. %Focus now on $p(\agen\given\estn)$, i.e., the probability of having an AoI value of \agen\ at a generic time instant $n$, conditioned on having at that time an estimate \estn. 
        Recalling the definition of \Irt, we observe that the AoI of the reference source is $0$ at the beginning of an inter-refresh period, and grows linearly over time, reaching the maximum value of $\Irt-1$ at the start of the last slot of the interval. Therefore, the probability for the receiver to have an AoI $\Agen=\agen$ at a generic time instant $n$ falling into an inter-refresh period of duration $\Irt(n)=\irt$ is simply $1/\irt$, $\forall$ $\agen\in\{0,\dots,\irt-1\}$. Leveraging this, the conditional PMF $p(\agen\given\estn)$ can be conveniently computed as
        \begin{align}
            \begin{split}
            \!\!\!\!\!p(\agen\given\estn) &\stackrel{(a)}{=} \sum\nolimits_{\irt > \agen} \frac{1}{\irt} \cdot \mathsf P\left(\Irt(n) = \irt \given \Estn=\estn\right)\!. %\\
                                %&\stackrel{(b)}{=} \sum_{\irt > \agen} \frac{p(\irt\given\estn)}{\mathbb E[\Irt \given \Estn=\estn]}.                        
            \end{split}
            \label{eq:condPMFAge}
        \end{align}
        Within \eqref{eq:condPMFAge}, $(a)$ follows from the law of total probability, observing that the AoI can reach a value $\agen$ only over an inter-refresh period of length at least $\agen+1$ slots, and using the uniform conditional probability for the AoI that was just derived. Plugging in \eqref{eq:condProbW} and recalling \eqref{eq:avgW} leads after simple steps to the proposition statement.
        %In turn, $(b)$ plugs into the expression the results in \eqref{eq:condProbW}. The proposition statement follows leaning on \eqref{eq:avgW}.                 
    \end{proof}
\end{prop}

\begin{prop}\label{prop:statEst}
    The stationary distribution of the receiver estimate for the reference source is given by    
    \begin{align}
        p(\estn) = \frac{\mathsf c_{\estn} \cdot \mathbb E[W \given \Estn=\estn]}{\mathsf c_{0} \cdot \mathbb E[W \given \Estn=0] + \mathsf c_{1} \cdot \mathbb E[W \given \Estn=1]}.
        \label{eq:pXn}
    \end{align}
    where 
    \begin{align}
        \mathsf c_0 = \frac{(\statZ (1-\qZO) \pTx_{00} + \statO \qOZ \pTx_{10})\ps}{\avgPTx \ps}, \quad  \mathsf c_1 = 1 - \mathsf c_0.
        \label{eq:c0}
    \end{align}
\end{prop}
\begin{proof}
    We start by observing that an inter-refresh period is characterized by having an estimate $\Estn=0$ with probability $\mathsf c_0$ reported in \eqref{eq:c0}. Here, the numerator captures the probability for the source to be in $0$ and to successfully deliver an update, i.e., $\statZ(1-\qZO)\pTx_{00}\ps$, or in $1$, transition to $0$ and inform the receiver, i.e., $\statO\qOZ\pTx_{10}\ps$. In turn, the denominator is a normalizing factor that accounts for the overall probability of delivering an update, i.e., initiating a new inter-refresh period. Similarly, the probability of having an inter-refresh interval with $\Estn=1$ is simply described by the auxiliary variable $\mathsf c_1 = 1-\mathsf c_0$. Leaning on this, the stationary distribution of $\Estn$ can be expressed as the fraction of time the receiver spends with such estimate value, obtaining the expression in \eqref{eq:pXn}.
\end{proof}

To conclude, the joint PMF $p(\agen,\estn)$ can be computed using \eqref{eq:condPMFAge_complete} and \eqref{eq:pXn}, eventually providing a complete analytical characterization of the conditional entropy $H(\Mcn\given\Agen,\Estn)$.

%To characterize this, we again resort to the auxiliary Markov process in \figr\ref{fig:abs_mc}, and  
%As a starting point, we observe that the chain can be used to capture the inter-refresh time of the estimate at the sink for the source of interest. Specifically, let us denote by \Irt\ the stochastic process describing the duration between two successive receptions of a message from the reference source, i.e., between two updates of the corresponding estimate at the sink. With this definition, the distribution of \Irt, conditioned on the period being characterized by an estimate value $\hat{x}$ at the receiver, is simply the absorption time for the chain when starting from $\hat{x}$.\footnote{Note indeed that counting the steps to absorption starting from state $i\in\{0,1\}$ is equivalent to assuming that at time $0$ a message is received, stating that the reference source is in such state, and thus starting a period over which the sink has an estimate $i$.} Accordingly, $p(\irt\given\estn)$ can be obtained using standard methods for Markov chains \cite{} as the discrete phase-type distribution
%\begin{align}
%    p(\irt\given\estn) = \mathbf e_{\estn}^{\mathsf T} \mathbf A^{\irt-1} \, \mathbf a_{\mathsf d}.
%\end{align}

\section{Results and Discussion}
\label{sec:results}

%To gauge the impact of different channel access policies on the uncertainty at the receiver, we will
We focus on three main strategies, labeled \emph{random}, \emph{reactive}, and \emph{balanced}. The first is obtained by setting $\pTx_{\mc_{n-1}\mcn} = 1/\nodes$, and corresponds to an approach that maximizes the throughput, transmitting regardless of the evolution of the tracked sources. In turn, the reactive scheme fully ties access to source transitions, sending an update if a state change occurs, and remaining silent otherwise, i.e., $\pTxVec=[0,1,1,0]$. Finally, we denote as balanced an approach in which \pTxVec\ has been optimized to minimize $H(\Mcn\given\Agen,\Estn)$.

As a preliminary step, the tightness of the approximation in \eqref{eq:ps} is verified in \figr\ref{fig:prob_verification}. Here, markers denote the calculation of \mbox{$\pr(\Mcn{=}0\given\agen,\Estn{=}0)$} and $\pr(\Agen{=}\agen,\Estn{=}0)$ derived via \eqref{eq:pmfXnGivenDeltaXnHat}, \eqref{eq:condPMFAge_complete}, \eqref{eq:pXn}, whereas solid lines are the outcome of simulations which jointly track the behavior of each node, characterizing the success probability exactly.

\begin{figure}
    \subfloat[]{
        \includegraphics[width=.49\columnwidth]{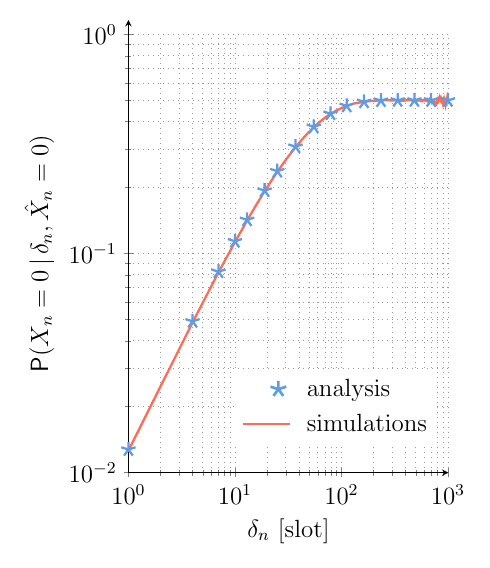}}
    \hspace*{.3em}
    \subfloat[]{
        \includegraphics[width=.49\columnwidth]{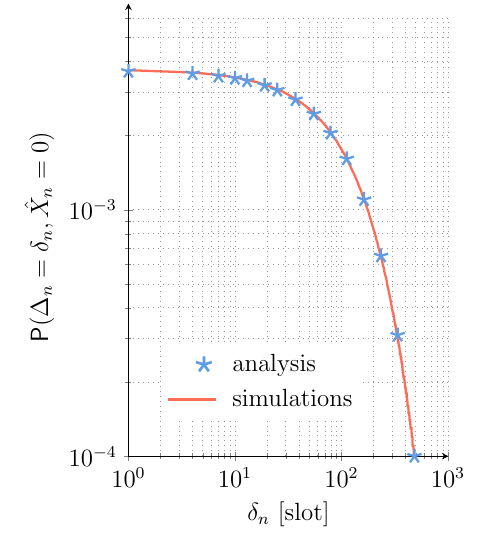}}
    \caption{Comparison of analysis (under the approximation \eqref{eq:ps}) and simulation (considering the exact interference behavior of the network). Results obtained for $\nodes=50$ nodes, $\qZO=\qOZ=0.02$, $\pTxVec = [0,1,1,0]$.}
    \label{fig:prob_verification}
\end{figure}

We then start by studying symmetric sources ($\eta=1$). Considering $\alpha{=}\beta{=}0.02$, Fig.~\ref*{fig:pmfH} reports the cumulative distribution function of the uncertainty at the receiver at a generic point in time, given the current value of AoI and estimate, obtained analytically as $\pr(\hspace{.1em}\condent(\agen,\estn){\leq} \zeta )= \!\!\sum_{\mathcal A} p(\agen,\est)$, where \mbox{$\mathcal A = \{(\agen,\estn) \!\!: \condent(\agen,\estn)\leq \zeta)\}$}.
%\begin{align}
%    \pr\big(\hspace{.1em}\condent(\agen,\estn)\leq \zeta \,\big)\,\,= \!\!\sum_{\substack{(\agen,\estn): \\ \condent(\agen,\estn)\leq \zeta}} \!\!\!p(\agen,\est).
%\end{align}
Clearly, the larger the number of nodes, the higher the uncertainty, as updates from a specific source are delivered more sporadically. More interestingly, the use of a reactive strategy reduces the receiver uncertainty. From this standpoint, although transmitting only in case of state change may leave the receiver with an incorrect estimate for longer times compared to randomly sending updates (see, e.g., \cite{Munari24_ICC}), the effect is more than counterbalanced by the beneficial impact of tying channel access to the source evolution. A solid intuition for this can be obtained considering the simple case $\nodes=1$. If a reactive strategy is employed, the receiver retains certain knowledge on source state as soon as the first message is received ($\condent(\agen,\estn) = 0$). Conversely, if the source transmits randomly, a change of state might have occurred without the receiver being notified, leading to an uncertainty that grows over time until a new update is received. The effect holds also as more nodes contend, with the random policy exhibiting a steeper rise over time of $\condent(\agen,\estn)$, as illustrated in the example of Fig.~\ref{fig:timeline_comparison}.

\begin{figure}
    \centering
    \includegraphics[width=.85\columnwidth]{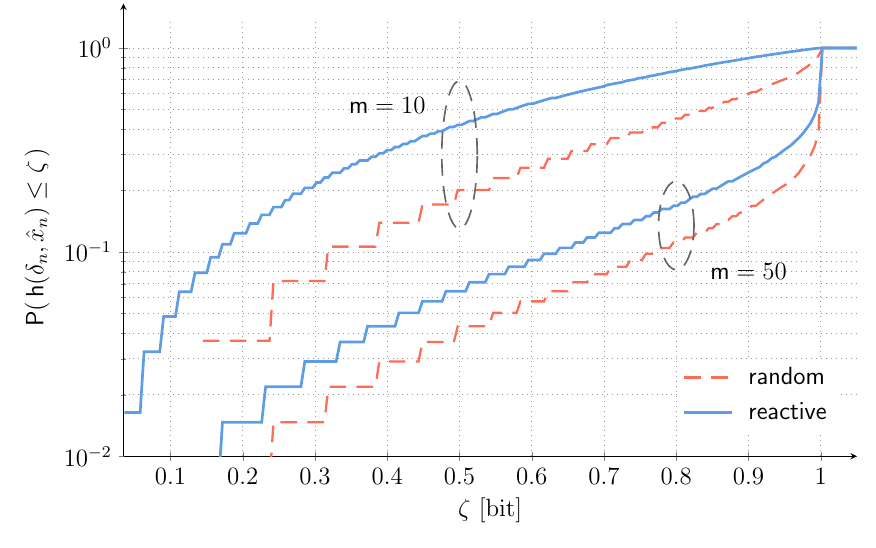}
    \vspace{-1em}
    \caption{CDF of $\condent(\agen,\est)$, considering symmetric sources with $\alpha=0.02$.}
    \label{fig:pmfH}
    \vspace{-1em}
\end{figure}

\begin{figure}
    \centering
    \includegraphics[width=.9\columnwidth]{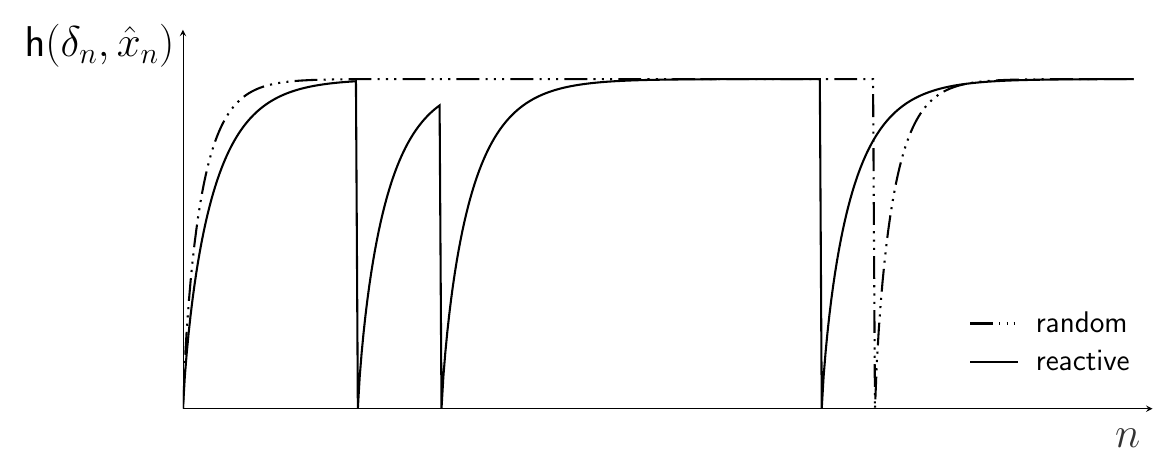}
    \caption{Example of time evolution of the entropy $\condent(\agen,\estn)$ for symmetric sources, with $\qZO=0.02$, $\nodes=50$. For the random strategy (dash-dotted lines), the transmission probability has been set to $\pTx=1/\nodes$. The steeper rise of the receiver uncertainty of this approach compared to the reactive one (solid line) is evident at the very beginning, as well as after the second reset of $\condent(\agen,\estn)$ for the random scheme.}
    \vspace{-1em}
    \label{fig:timeline_comparison}
\end{figure}

To further elaborate, \figr\ref*{fig:symmetric}a reports $H(\Mcn\given\Agen,\Estn)$ as a function of the number of nodes. As discussed, for large populations, the average uncertainty converges for all policies to the stationary entropy of the source $\mathsf H(\Mc)$. Focus first on the random and balanced schemes. For the latter, the optimized transmission probabilities are shown by the circle-marked line in \figr\ref*{fig:symmetric}b, referring to the left $y$-axis, and compared for reference to the $\lambda{=}1/\nodes$ access probability of the random scheme. The numerical optimization leads to $\pTx_{00}{=}\pTx_{11}{=}0$ regardless of \nodes, i.e., a node shall send no update if the source does not transition. On the other hand, a pure reactive solution is optimal ($\pTx_{01}{=}\pTx_{10}{=}1$) as long as the channel is not congested. This is confirmed by the green dash-dotted curve, referred to the right $y$-axis, showing the  average number of incoming packets per slot at the receiver (i.e., channel load). As \nodes\ increases, having nodes transmit only in case of state change, yet with probability smaller than $1$, is convenient to curb collisions. These  trends pinpoint a trade-off between throughput (maximized for load $1$ pkt/slot) and uncertainty at the receiver, with the latter improved by transmitting fewer yet more informative packets. This is further buttressed by the yellow, dash-dotted curve in \figr\ref*{fig:symmetric}, outlining a modification of the plain reactive scheme in which nodes transmit also in case of no source change, with probability chosen such that the overall average channel load is always $1$ pkt/slot, to maximize throughput. An increase in the receiver uncertainty is apparent.

\begin{figure}
    \subfloat[]{
        \includegraphics[width=.48\columnwidth]{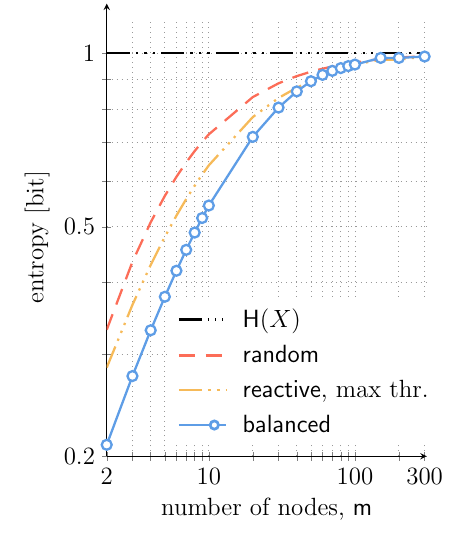}}
    \hspace*{.1em}
    \subfloat[]{
        \includegraphics[width=.5\columnwidth]{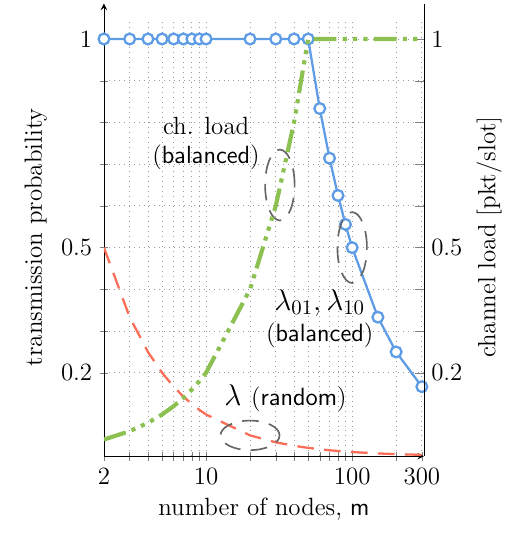}}
    \caption{(a) $\ent(\Mcn\given\Agen,\Estn)$ vs number of nodes, symmetric sources ($\qZO=0.02$); (b) transmission probabilities (random and balanced strategies).}
    \label{fig:symmetric}
    \vspace{-1em}
\end{figure}

\begin{figure}
    \subfloat[]{
        \includegraphics[width=.45\columnwidth]{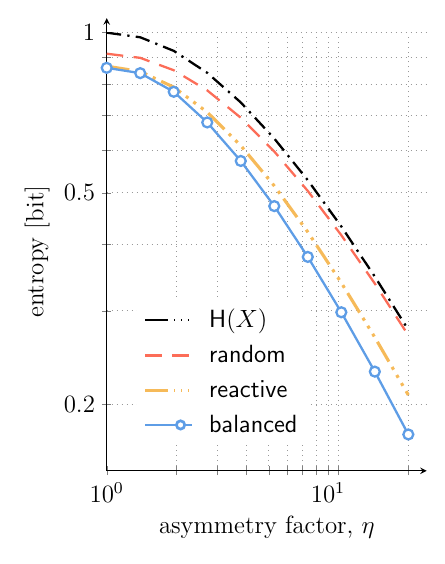}}
    \hspace*{.1em}
    \subfloat[]{
        \includegraphics[width=.47\columnwidth]{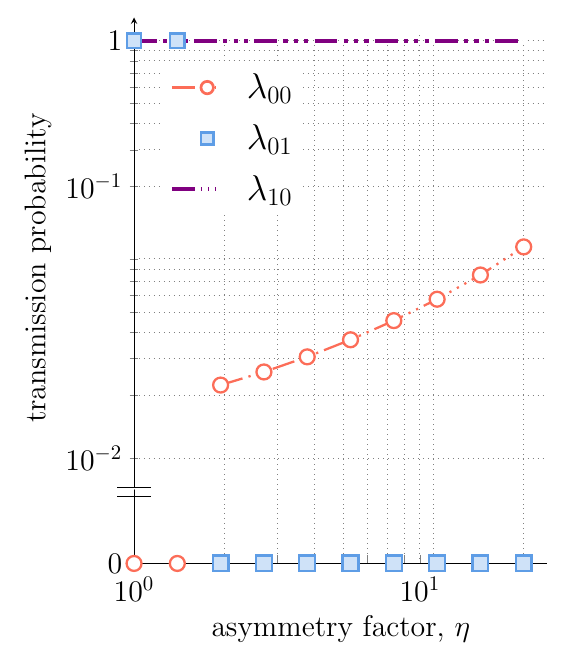}}
    \caption{(a) $\ent(\Mcn\given\Agen,\Estn)$ vs asymmetry factor $\eta$; (b) optimized transmission probabilities for the balanced strategy. $\nodes=50$; source parameters set to have on average $\nodes(\pi_0\alpha {+} \pi_1\beta){=}0.8$ state transitions per slot in the network.}
    \label{fig:asymm}
    \vspace{-1em}
\end{figure}

The situation changes significantly when asymmetric sources are to be monitored. To study this, we report in \figr\ref*{fig:asymm}a $H(\Estn\given\Agen,\Estn)$ against the asymmetry factor $\eta$ for $\nodes{=}50$ nodes. The source parameters $(\alpha$, $\beta)$ have been set such that, for any $\eta$ the average number of state transitions in the network is constant, i.e., $\nodes(\pi_0\alpha {+} \pi_1\beta){=}0.8$. As expected, as $\eta$ increases, sources spend more time in one state (i.e, $1$), leading to lower uncertainty. Two facts shall be stressed, though. First, resorting to a pure reactive policy is once more be preferred to the use of a random solution, all the more so with high asymmetry. However, in contrast to what discussed for the symmetric case, a solely reactive approach no longer minimizes the uncertainty, as done by the balanced scheme. This is clarified in \figr\ref{fig:asymm}b, reporting the optimized transmission probabilities ($\lambda_{11}{=}0$ in all cases, is not shown for clarity). To reduce uncertainty, a transition to the less visited state $0$ shall be immediately notified ($\pTx_{01}{=}1$). As asymmetry grows, updates shall also be sent when the source remains in such state ($\pTx_{00}$ rising with $\eta$), compensating for a possible loss of the initial notification due to collision, and informing the receiver of the less likely conditions being experienced. Such messages become more important than notifying a return to the most common state, and $\lambda_{01}$ goes to $0$ to reduce congestion.

\section{Conclusions}
For remote monitoring of two-states Markov sources over random access channels, we characterized analytically the uncertainty of a simple estimator which retains the last received value as current estimate. Our study highlighted important trade-offs, showing that reactive schemes are to be preferred in the case of symmetric sources, even at the expense of operating the system at lower throughput. Instead, when asymmetric sources are to be tracked, a hybrid solution that also foresees transmissions in absence of a state transition can lead to lower entropy. In all cases, the presented framework allows for a simple optimization, providing useful protocol design hints.

%\clearpage
%\newpage
\bibliographystyle{IEEEtran}
\bibliography{IEEEabrv,biblio_RandomAccess,biblio_AoI}

\flushend

\end{document}